\documentclass[conference,letterpaper]{IEEEtran}

\usepackage[utf8]{inputenc} 
\usepackage[T1]{fontenc}
\usepackage{url}
\usepackage{ifthen}
\usepackage{cite}
\usepackage{graphicx}
\usepackage[cmex10]{amsmath} 

\usepackage[acronyms, nonumberlist, nopostdot, nomain, nogroupskip]{glossaries}
\usepackage{amsmath}
\usepackage{amsthm}
\usepackage{amsfonts}
\usepackage{bbm}
\usepackage{xcolor}

\newtheorem{theorem}{Theorem}
\newtheorem{lemma}{Lemma}

\newtheorem{definition}{Definition}
\newtheorem{remark}{Remark}

\newacronym{ldp}{LDP}{large deviation principle}
\newacronym{mdp}{MDP}{Markov decision process}

\newcommand{\advhyp}{\set H ^\text{adv}}

\newcommand{\set}{\mc}
\newcommand{\thetaadv}{\theta^{\text{adv}}}
\newcommand{\Padv}{P^{\text{adv}}}
\newcommand{\Tadv}{T^{\text{adv}}}
\newcommand{\mc}{\mathcal}

\DeclareMathOperator{\rank}{rank}

\begin{document}
\title{Covert Adversarial Actuators in Finite MDPs} 


\author{%
  \IEEEauthorblockN{Edoardo David Santi, Gongpu Chen, Deniz Gündüz}
  \IEEEauthorblockA{Department of Electrical and Electronic Engineering \\
                    Imperial College London\\
                    London, UK\\
                    \{eds17, gongpu.chen, d.gunduz\}@ic.ac.uk}
  \and
  \IEEEauthorblockN{Asaf Cohen}
  \IEEEauthorblockA{The School of Electrical and Computer Engineering\\
                    Ben-Gurion University of the Negev\\ 
                    Beer-Sheva, Israel\\
                    coasaf@bgu.ac.il}
}

\maketitle


\begin{abstract}
   We consider a \gls{mdp} in which actions prescribed by the controller are executed by a separate actuator, which may behave adversarially. At each time step, the controller selects and transmits an action to the actuator; however, the actuator may deviate from the intended action to degrade the control reward. Given that the controller observes only the sequence of visited states, we investigate whether the actuator can covertly deviate from the controller’s policy to minimize its reward without being detected. We establish conditions for covert adversarial behavior over an infinite time horizon and formulate an optimization problem to determine the optimal adversarial policy under these conditions. Additionally, we derive the asymptotic error exponents for detection in two scenarios: (1) a binary hypothesis testing framework, where the actuator either follows the prescribed policy or a known adversarial strategy, and (2) a composite hypothesis testing framework, where the actuator may employ any stationary policy. For the latter case, we also propose an optimization problem to maximize the adversary’s performance.
\end{abstract}
\section{Introduction}
%
%
There are many scenarios in which an adversary can be stopped, or would suffer negative consequences if identified. Hence, the adversary may wish to limit its adversarial behaviour to stay covert. Conversely, there are scenarios where an agent wishes to complete a task while staying covert so that an adversary cannot interfere or extract information from it.
A setting of this type is the problem of covert communication \cite{towsley, lizhong, block_covert_comms}. It involves designing communication protocols that enable a sender to transmit messages to a receiver without being detected by an adversary monitoring the communication channel. Adversaries in control systems extend these ideas by interfering with the operation of controllers to make their tasks harder. The goal is to influence the state or behaviour of a dynamic system (e.g., a robot, autonomous vehicle, or networked system) while remaining stealthy: avoiding anomalies that might trigger alarms or reveal the presence of external interference. Attacking control systems also has to account for their dynamic and feedback-driven nature, making it a more complex task. 

\begin{figure}[t] \label{fig}
	\centering
	\includegraphics[width=1.8in]{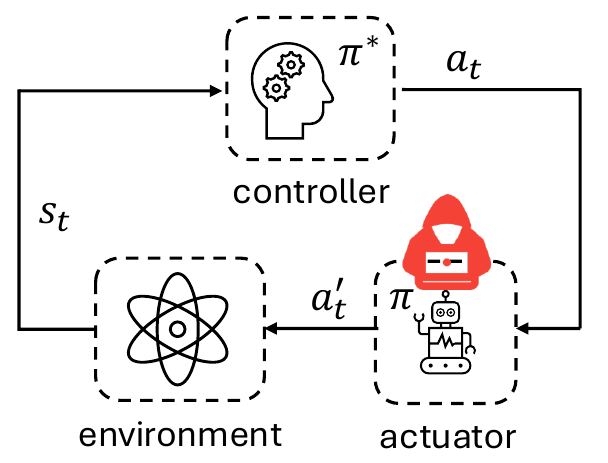}
	\caption{System model. The controller wants the actuator to follow policy $\pi^*$; however, the compromised actuator instead aims to minimize the reward without being detected.}
	\label{fig:model}	
 \vspace{-0.5em}
\end{figure} 

One type of attack involves corrupting the observations of the monitor or the controller of the system to distort its knowledge about the system state, and to derail its actions \cite{mo2009secure, smith2011decoupled, mo2010false, guo2017optimal, 7563349}. In another class of attacks, the adversary corrupts the control signal \cite{bai2015security, bai2017data,zhang2016stealthy,kung2017performance,weerakkody2016information,asaf1,asaf2}. 
In \cite{chang2023distributed}, the problem of command corruption is studied in multi-armed bandits, where the arms pulled are different from those commanded. In other settings, an agent may wish to stay covert while completing a task. In \cite{blochcovertsensing}, the authors consider an active sensing scenario, where the goal of the adversary is to estimate a parameter without being detected. In \cite{chang2022covert}, covert best arm detection is studied in a multi-armed bandit setting. In \cite{ufuk}, the authors study a reachability problem in an \gls{mdp}, where the objective is to limit the probability of an adversary observing the states to infer their transition probabilities under the chosen policy.

In this paper, we consider a controller interacting with the environment by controlling the actions of an actuator, formulated as an \gls{mdp}. The adversary either takes over the actuator itself, or infiltrates the communication channel between the controller and the actuator, corrupting the instructions received by the actuator. The adversary aims to modify controller's policy to degrade the system's performance, e.g., to minimize the long-term average reward, while avoiding detection. See Fig.~\ref{fig} for an illustration of the system model. We assume that the system dynamics, i.e., the state transition probabilities and reward function, and the current state are known by both the controller and the adversarial actuator, while the average reward is learned only at the end. Hence, the goal of the adversary is to corrupt the actions taken without significantly distorting the statistics of the visited states to avoid detection.

\noindent \textbf{Notation:} We denote random variables with upper case letters and realizations with the associated lower case letters. For $n,k \in \mathbbm{Z}^+$, we use $s_n^{n+k}$ to represent the sequence $\{s_n,s_{n+1},\ldots,s_{n+k}\}$. $s_1^n$ will be abbreviated as $s^n$ for simplicity and 
for any symbol, `$_{s^n}$' subscript indicates an empirical estimate using sequence $s^n$. For any set $\mathcal{X}$, $\Delta(\mathcal{X})$ denotes the set of probability distributions over $\mathcal{X}$, while $\mathcal{X}^0$ denotes its interior and $\overline{\mathcal{X}}$ its closure, both in the total variation metric. $\mathbbm{1}\{\cdot \}$ is the indicator function. For two distributions $p$ and $q$ over $\set{X}$ we define the relative entropy as $H(p\|q)=\sum_{i \in \set{X}}p(i)\log\frac{p(i)}{q(i)}$. We adopt $\log=\log_2$. 
\section{Problem Formulation}
\subsection{System model}
An \gls{mdp} is defined by the tuple $(\set{S},\set{A},T,r,\mu)$, where $\set{S}$ and $\set{A}$ are the state and action spaces,  respectively, $T:\set{S}\times \set{A}\times \set{S}\to [0,1]$ is the transition kernel, $r:\set{S}\times \set{A}\to \mathbb{R}$ is the reward function, and $\mu \in \Delta(\set{S})$ is the initial state distribution. In this work, we focus on \gls{mdp}s with finite state and action spaces.
At every time $t$, the controller selects an action $a_t\in \mathcal{A}$ based on the current state $s_t$ and transmits $a_t$ to the actuator. Upon the execution of $a_t$ by the actuator, the environment generates a reward $r(s_t,a_t)$ and transitions to the next state $s_{t+1}$ following the distribution $T(\cdot|s_t,a_t)$. We assume that the controller cannot observe the instantaneous reward at each time step; instead, only the average reward is revealed at the very end.
A stationary policy for the MDP is a mapping that maps the current state to a distribution over actions. We denote by $\pi(a_t|s_t)$ the probability that policy $\pi$ selects action $a_t$ in state $s_t$. The expected average reward of policy $\pi$ over an infinite time horizon is defined as
\begin{align*}
    J(\pi) := \lim_{N\to\infty} \frac{1}{N} \mathbb{E}_{\pi} \left[\sum_{t=1}^N r(s_t,a_t)|s_1\sim \mu \right].
\end{align*}
Let $T_\pi$ denote the transition matrix induced by policy $\pi$, where each entry is given by $T_\pi(s,s')=\sum_{a} T(s'|s,a)\pi(a|s)$.
Throughout this paper, we assume that the MDP is recurrent, meaning that the Markov chain induced by any stationary policy consists of a single recurrent class. As a result, under any stationary policy $\pi$, the system reaches a unique stationary state distribution, denoted by $\Bar{\theta}_\pi$. This stationary distribution satisfies $\tau_{\pi}^{\top}T_\pi=\tau_{\pi}^{\top}$.
Moreover, let $\theta_\pi (s,s')=\tau_\pi (s) T_\pi (s,s'),\;\forall s,s' \in \set{S}$, be the stationary distribution over state transitions induced by policy $\pi$. 

This work investigates a scenario where the controller wants to follow a predefined policy $\pi^*$, while the adversarial actuator may not faithfully execute the actions prescribed by $\pi^*$. Fig. \ref{fig:model} provides an illustration of the system. For simplicity, we refer to the actuator and the adversary interchangeably as a single entity. We assume that the adversary has full knowledge of the control policy $\pi^*$ and that both the controller and the adversary have access to the environment state at every time. Consequently, at every time $t$, the controller prescribes an action $a_t$ using policy $\pi^*$; however, the adversary may override this action and instead execute an action $a_t'$, selected according to an adversarial policy $\pi$. To quantify the impact of adversarial behavior, we define the regret induced by the adversarial policy $\pi$ as:
\begin{align*}
    R(\pi) = J(\pi^*) - J(\pi)~.
\end{align*}
The objective of the adversary is to find a policy $\pi$ that maximizes the regret covertly, as being found out might result in negative repercussions. On the other hand, the controller aims to detect with the best possible accuracy whether the actuator is acting faithfully or not.
%
%
%

\subsection{Detection and Covertness}
Since the controller has full access to the entire state sequence up to the current time, it can leverage this information to detect any abnormal behavior based on observed state transitions. In this part, we discuss the controller’s ability to detect whether the control policy $\pi^*$ is executed faithfully.

We formulate the controller's detection problem as a hypothesis test, where the null hypothesis $\set H ^*$ corresponds to the actuator acting faithfully to the controller's policy, i.e. $\pi=\pi^*$, while the alternative hypothesis $\advhyp$ corresponds to adversarial behaviour, i.e. $\pi \neq \pi^*$. For any symbol, let `$^*$' and `$^{\text{adv}}$' superscripts indicate belonging to the two hypotheses, respectively. Let $g_n:\mathcal{S}^n\to \{0,1\}$ denote the controller's decision function at time $n$, where $g_n(s^n)=0$ means that $\set H ^*$ is accepted at time $n$, and $g_n(s^n)=1$ means that $\advhyp$ is accepted. Let $P^*_n:\set{S}^n \mapsto [0,1]$ and $\Padv_n:\set{S}^n \mapsto [0,1]$ be the probabilities of sequences of length $n$ under the null and alternative hypothesis, respectively. Define $\mathcal{B}_n \subseteq \mathcal{S}^n$ as the set of sequences for which $g_n(s^n)=0$, and denote its complement by $\mathcal{B}_n^c$, where $g_n(s^n)=1$. Then the two probabilities of error are defined as follows:
\begin{align*}
    \alpha_n :=& \Pr\{g_n(s^n)=1|\set H ^* \text{ true}\} = P^*_n(\mathcal{B}_n^c )~,   \\
    \beta_n :=& \Pr\{g_n(s^n)=0| \advhyp \text{ true}\} = \Padv_n(\mathcal{B}_n )~.
\end{align*}
Specifically, let $\alpha=\lim_{n\to \infty} \alpha_n$ and $\beta=\lim_{n\to\infty} \beta_n$. We assume the controller performs the optimal detection algorithm, defined as the one minimizing the total probability of error $\alpha+\beta$. An adversarial policy is said to be $\epsilon$-covert if $\alpha+\beta=1-\epsilon$ under this policy.
For any recurrent \gls{mdp} with finite state and action spaces, for any non-stationary policy there exists a stationary policy that achieves the same long-term average state-action frequencies \cite{tsitsiklis}. We show later that the covertness properties of a policy depend on its long-term state transition frequencies, which are fully determined by its state-action frequencies, which also fully determine the long-term average reward. Hence, we can conclude that for any non-stationary policy, there exists a stationary policy that achieves the same average long-term reward and covertness, and so we can limit ourselves to considering stationary policies. 
\section{Perfect Covertness over an Infinite Horizon} \label{sec:infinity}
In this section, we analyze the covertness of the adversarial policy over an infinite time horizon. In this setting, an adversarial policy can either be 0-covert—meaning it will be detected by the controller with probability 1—or 1-covert, indicating perfect covertness. We derive the necessary and sufficient conditions for a policy to be 1-covert, and subsequently formulate the problem of finding the optimal adversarial policy as a linear program. 

For any state sequence $s^n$, we define the following empirical distributions: for any $s,s'\in \mathcal{S}$,
\begin{align}  \label{eq:emp-tau}
    &\tau_{s^n}(s) = \frac{1}{n} \sum_{t=1}^{n} \mathbbm{1}\{s_t=s\},  \\  \label{eq:emp-theta}
    &\theta_{s^n}(s,s') = \frac{1}{n-1} \sum_{t=1}^{n-1} \mathbbm{1}\{s_t=s,s_{t+1}=s' \}, \\ \label{eq:emp-T}
    &T_{s^n}(s,s') = \frac{\theta_{s^n}(s,s')}{\tau_{s^n}(s)}.
\end{align} 
Suppose that the adversary adopts a stationary policy $\pi^{\text{adv}}$. Denote by $T^*$ and $\Tadv$ the transition matrices induced by polices $\pi^*$ and $\pi^{\text{adv}}$, respectively. Recall that we assume the MDP is recurrent, hence both $T^*$ and $\Tadv$ are irreducible.

\begin{theorem}\label{infinite horizon necessary sufficient condition}
For any stationary policy $\pi^{\text{adv}}\neq\pi^*$, $\pi^{\text{adv}}$ is 1-covert if $\Tadv=T^*$ and 0-covert otherwise.
\end{theorem}
\begin{proof}
     As $n\to\infty$, $T_{s^n} \to \Tadv$. Thus, if $T_{s^n} \neq T^*$, it is obvious that $\pi^{\text{adv}} \neq \pi^*$, the controller accepts $\advhyp$ and $\epsilon=1$. If $T_{s^n} = T^*$, and there exist multiple stationary policies inducing the transition matrix $T^*$, the controller is uninformed and $\epsilon=0$. The case where there exists a single stationary policy inducing the transition matrix $T^*$ is degenerate as it implies that $\pi^{\text{adv}}=\pi^*$, in which case the detection is between two identical hypotheses.
\end{proof}
Theorem \ref{infinite horizon necessary sufficient condition} implies that an adversarial policy $\pi^{\text{adv}}$ is perfectly covert if and only if its transition matrix is identical to that of the original control policy. Denote by $\Pi_1$ the set of policies satisfying this condition, i.e., $\Pi_1 := \{\pi^{\text{adv}}: \Tadv = T^*  \}$. Let $T_s$ be the $|\set{A}| \times |\set{S}|$ matrix with $T_s(a,k)=T(k|s,a)$. Then $\pi\in \Pi_1$ if and only if $T_s^\top (\pi(\cdot|s) - \pi^*(\cdot|s))=0$ for all $s\in \mathcal{S}$. Hence $\Pi_1$ is convex whenever it is non-empty.

Given Theorem~\ref{infinite horizon necessary sufficient condition}, we can set an optimization problem which finds the optimal adversarial policy from the adversary's point of view, with the constraint that $\epsilon = 0$. Define the $(|\set{S}|+1)\times|\set{A}|$ matrix 
$
    C=
    \begin{bmatrix}
        T_s &
        \mathbf{1}
    \end{bmatrix}^\top
$, and $\Delta \pi = \pi - \pi^*$. Then, for each $s \in \set{S}$, the best adversarial policy that guarantees $\epsilon=0$ is given by
\begin{equation} \label{prob: infty}
    \begin{aligned}
    \min_{\Delta \pi(\cdot|s)} \;&\Delta \pi(\cdot|s)^\top r(s,\cdot) \\
    s.t.\; &C \Delta \pi(\cdot|s) = \mathbf{0}, \\
    &-\pi^*(s,a)\leq \Delta\pi(s,a)\leq1-\pi^*(s,a), \forall a \in\set{A} 
\end{aligned}
\end{equation}
We can see that the admissible values of $\Delta \pi$ belong to the null space of $C$. Thus, if this matrix has the full column rank ($\rank(C)=|\set{A}|$), the null space is trivial, containing only the origin, and the adversary cannot change the policy covertly. Meanwhile, if $\rank(C)<|\set{A}|$, the dimensionality of the space of admissible policies is $|\set{A}|-\rank(C)$.

\begin{remark}
Clearly, if we have two actions, a good one, $a_g$ and a bad one $a_b$ such that for all $s,s'$ we have $T(s'|s,a_g)$ =  $T(s'|s,a_b)$, but for some $s$ we have $r(s,a_g) > r(s,a_b)$, then the adversary can change the policy and increase the regret, without being noticed by the controller who sees the complete state sequence. So there exist MDPs for which the solution to problem (\ref{prob: infty}) is not only $\pi^*$ but also a strictly suboptimal policy. This result shows that the adversary can also achieve this in more complex ways, as in each state, it can modify its policy as long as this does not affect the transition probabilities between states. See Appendix~\ref{examples_infinity} for examples.
\end{remark}
\section{Asymptotic error exponents} \label{sec:asymptotic}
In the previous section, we established that as the length of the observed state sequence approaches infinity, the total probability of detection error, $\epsilon$, converges to 1 whenever the adversarial policy induces a transition matrix different from that of the controller's policy. However, in practical scenarios, the time horizon is finite, albeit potentially large. This raises an important question: how rapidly does $\epsilon$ increase as a function of time for an adversarial policy $\pi^{\text{adv}}$ when $\Tadv\neq T^*$? To address this, we focus on deriving the asymptotic error exponent, which characterizes the rate at which the probability of error decays over time.

We derive the asymptotic results using tools from large deviation theory. Suppose $\{ S_t \}$ is the Markov chain induced by a stationary policy $\pi$. Then, given a realization $s^n$ of the state sequence, we can construct the sequence of empirical estimates $ \{ \theta_{s^n} \} $ using \eqref{eq:emp-theta}. Let $P_\pi(\cdot)$ denote the probability distribution of $\theta_{S^n}$ under policy $\pi$.
%
\begin{definition}
    The process $ \{ \theta_{S^n} \}$ satisfies the \gls{ldp}\footnote{We define the \gls{ldp} in a stricter sense than in \cite{alma9910123510001591}.} with rate function $I$ if, for every Borel subset $\Gamma \subseteq \Delta(\set{S}^2)$, we have:
\begin{equation*}
    \lim_{n \to \infty} \frac{1}{n} \log P_\pi \{ \theta_{S^n} \in \Gamma\} = 
- \inf_{\nu \in \overline{\Gamma}} I(\nu).
\end{equation*}
where $I:\Delta(\set{S}^2) \to \mathbb{R} \cup \{ \infty \}$ is a continuous mapping, and $\Gamma$ does not have any isolated points, i.e. $\Gamma \subseteq \overline{\Gamma^0}$. 
\end{definition}

Intuitively, satisfying the \gls{ldp} means that given a large number of samples, the probability of the sequence of realizations of the process being within a certain set decreases exponentially at a rate depending on the member of the set giving the slowest rate of decrease.
Let $\set{M}$ be the set of shift-invariant measures, meaning that $\set{M}=\{\nu \in \Delta(\set{S}^2):\sum_{j \in \set{S}}\nu_{i,j}=\sum_{j \in \set{S}}\nu_{j,i}\}$. Let $\theta^1, \theta^2 \in \set{M}$ be two shift-invariant distributions over state transitions, $\tau^1, \tau^2$ be the corresponding  marginal distributions over states, and $T^1,T^2$ be the corresponding transition matrices. Then, we can define the differential divergence $D_K : \mc{M}^2 \mapsto \mathbb{R}^{+}_0 \cup \{\infty\}$ between state transition distributions:
\begin{align}
    D_K(\theta^1,\theta^2) := \sum_{s\in \mathcal{S}}\tau^1(s)\sum_{s'\in\mathcal{S}} T^1(s,s')\log \frac{T^1(s,s')}{T^2(s,s')}.
\end{align}
According to \eqref{eq:emp-tau}-\eqref{eq:emp-T}, it is easy to verify that $D_K(\theta^1,\theta^2)=H(\theta^1||\theta^2) - H(\tau^1||\tau^2)$. The following result from the large deviation theory is useful in our analysis:

\begin{theorem} \label{theorem: sanov}
\cite[Theorem~3.1.13]{alma9910123510001591} Given a Markov process $\{ S_t \}$ with irreducible transition matrix $T$ and state transition frequencies $\theta$, the $\Delta(\set{S}^2)$-valued process $\{ \theta_{S^n} \}$ satisfies the \gls{ldp} with the rate function 
\begin{equation*}
    I(\nu)
    =\begin{cases}
        D_K(\nu,\theta), & \text{if } \nu \in \set{M} \\
        \infty, & \text{if }\nu \notin \set{M}
    \end{cases}
\end{equation*}
\end{theorem}
Note that this result is a variation of the classic result given by Sanov's theorem, which is defined for i.i.d. processes instead and has a rate $I(\nu)=H(\nu \| \phi)$, where $\phi$ is the distribution of the i.i.d. variables. 
\subsection{Predefined adversarial policy} \label{sec:fixed hypothesis}
We start from a simple setting where the adversary either follows the control policy faithfully or deviates by executing a known adversarial policy $\pi^{\text{adv}}$, which is known to the controller as well. In this setting, the controller is faced with a binary hypothesis testing problem. The null hypothesis, $\mathcal{H}^*$, represents the scenario where the actuator adheres to the intended control policy, while the adversarial hypothesis, $\advhyp$, corresponds to the scenario where the actuator executes $\pi^{\text{adv}}$. Both policies are assumed to be stationary and known to both the controller and the adversarial actuator. 

Define the normalized log-likelihood ratio of $s^n$ as
\begin{equation}\label{eq:log-likelihood}
    \begin{aligned} 
    L(s^n)&=\frac{1}{n-1}\log \frac{P^*_n(s^n)}{\Padv_n(s^n)}=D_K(\theta_{s^n}, \thetaadv) - D_K(\theta_{s^n}, \theta^*),
\end{aligned}
\end{equation}
where the second equality is shown in Appendix~\ref{app:log-likelihood}. We set our detector $g_n$ as the normalized log-likelihood ratio test, which is equivalent to the likelihood ratio test used in binary hypothesis testing. It accepts the null hypothesis ($g_n(s^n)=0$) when $L(s^n)>\eta$, where $\eta$ is a threshold, and it accepts the alternative ($g_n(s^n)=1$) otherwise.
The normalized log-likelihood ratio test defines a set $\set{B}_n \subseteq \set{S}^n$, under which the null hypothesis is accepted, and its complement $\set{B}_n^c$, under which the alternative hypothesis is accepted. We can see from (\ref{eq:log-likelihood}) that the normalized log-likelihood ratio depends on the trajectory $s^n$ only through the empirical distribution over state transitions $\theta_{s^n}$. Thus, the test equivalently defines a set $\hat{\set{B}}_n \subseteq \Delta(\set{S}^2)$, in which $\set{H}^*$ is accepted, and its complement $\hat{\set{B}}_n^c$, in which $\advhyp$ is accepted, depending on the location of $\theta_{s^n}$. 

The Neyman-Pearson lemma states that for a sequence of random variables observed, the optimal detector uses a threshold on the likelihood ratio. The null hypothesis is rejected when the likelihood ratio is lower than a threshold, and it is accepted otherwise. This is optimal in the sense that it is not possible to decrease the probability of type I error (i.e., $\alpha_n$) without increasing the probability of type II error (i.e., $\beta_n$). The threshold can be adjusted higher or lower based on the type of error we prioritize reducing.
\subsubsection{Error exponents for a fixed threshold} 
In this part, we analyze the asymptotic decay of $\alpha_n$ and $\beta_n$ when setting a fixed threshold for all $n$.
For a fixed threshold $\eta$, the sets $\hat{\set{B}}_n$ are the same for any $n\geq 2$, i.e., $\hat{\set{B}}_n=\hat{\set{B}}$. Then, we can define the two probabilities of error for all $n\geq 2$ as 
\begin{equation*}
    \alpha_n= \Pr(\theta_{s^n} \in \hat{\set{B}}^c | \set{H}^*), \;\;\;
    \beta_n= \Pr(\theta_{s^n} \in \hat{\set{B}} | \advhyp).
\end{equation*}
Assuming irreducibility of the hypotheses' transition matrices, which is guaranteed by the recurrent MDP assumption, applying Theorem~\ref{theorem: sanov} to these equations, noting that $I$ in this form is continuous and that both $\hat{\set{B}}$ and $\hat{\set{B}}^c$ do not contain isolated points (see Appendix~\ref{app: no isolated points} for the proof), we obtain the following. 
\begin{theorem}
    Under the normalized log-likelihood ratio detector with a fixed threshold $-D_K(\thetaadv, \theta^*) \leq \eta \leq D_K(\theta^*, \thetaadv)$, the asymptotic error exponents are given by
    \begin{align*}
    \lim_{n \to \infty} \frac{1}{n} \log \alpha_n&= -\inf_{\theta \in \overline{\hat{\set{B}}^c} \cap \mc{M}} D_K(\theta, \theta^*), \\
    \lim_{n \to \infty} \frac{1}{n} \log \beta_n&= -\inf_{\theta \in \overline{\hat{\set{B}}} \cap \mc{M}} D_K(\theta, \thetaadv) .
\end{align*}
where the infimum of the empty set is taken to be $\infty$. 
%

\end{theorem}
Note that the threshold must take those values as to have a boundary between $\hat{\set{B}}_n$ and $\hat{\set{B}}_n^c$ that is located between the two hypotheses in the space of state transition frequencies. Informally speaking, $D_K$ can be considered a divergence between different distributions over state transitions. In general, the divergences in the above statement grow if the divergence between $\theta^*$ and $\thetaadv$ is greater, leading to higher error exponents due to the processes becoming increasingly different, and thus, easy to distinguish. Additionally, note that when the exponents are $0$, the error probabilities do not decay asymptotically. This occurs if $\theta^*=\thetaadv$, meaning that the two Markov chains are indistinguishable, which is the setting described in Section~\ref{sec:infinity}.
\subsubsection{Error exponent for a fixed $\alpha_n$ (error type I)} In the following, we derive the asymptotic error exponent of $\beta_n$ (type II error) when instead of fixing the threshold of the detector, we fix an upper bound on the type I error probability and adjust the threshold accordingly. We now provide a version of the Chernoff-Stein lemma, adapted to the Markov chain scenario.
\begin{theorem} \label{theorem:chernoff}
    Consider the binary hypothesis test when $s^n$ is a Markov chain drawn according to either of the two state transition distributions $\theta^*$ and $\thetaadv$, respectively, where $D_K(\theta^*,\thetaadv)<\infty$. Let $\set{A}_n \subseteq \set{S}^n$ be an acceptance region for the null hypothesis $\mc{H}^0$. Let the probabilities of error be
    \[
    \alpha_n=P^*_n(\set{A}^c_n), \;\;\;\;\;\; \beta_n=\Padv_n(\set{A}_n)
    \]
    and for $0< \delta < \frac{1}{2}$, define
    \[
    \beta_n^{\delta}=\min_{\substack{\set{A}_n \subseteq \set{S}^n \\ \alpha_n < \delta}} \beta_n.
    \]
    Then
    \begin{equation} \label{eq: thrm3 main result}
    \lim_{n\to \infty} \frac{1}{n} \log \beta_n^\delta = - D_K(\theta^*,\thetaadv).
    \end{equation}
\end{theorem} 

\begin{proof}
    The detailed proof is provided in Appendix~\ref{app: proof 4 properties}. 
\end{proof}
Again, it can be seen that the rate of decay of the error probabilities is improved as the divergence between $\theta^*$ and $\thetaadv$ increases, with an error exponent of zero when $\theta^*=\thetaadv$.
\subsection{Unknown adversarial policy}

In Section~\ref{sec:fixed hypothesis}, we analyzed the performance of a detector with the assumption that the alternative hypothesis $\advhyp$ is known. In this section, we study the setting where the controller has no knowledge of the adversarial policy or its possible distribution, and thus conducts the detection in an `anomaly detection' fashion: for any state sequence, the detector calculates a statistic of the sequence and compares it with a fixed threshold. If the statistic is above the threshold, the detector does not reject the null hypothesis, otherwise it is rejected. The threshold is chosen depending on the accepted level of Type I error, i.e., the probability of wrongly rejecting the null hypothesis. To decide on the best statistic to use, we define optimality as initially proposed by Hoeffding \cite{alma9910123510001591}: 
\begin{definition}
    A test $\set{S}$ is optimal (for a given threshold $\eta > 0$) if, among all tests that satisfy 
    \begin{equation*} \label{eq: alpha constraint}
    \limsup_{n\to\infty} \frac{1}{n} \log \alpha_n \leq -\eta,
    \end{equation*}
    test $\set{S}$ has the maximal exponential rate of error, i.e., uniformly over all possible laws $\mu_1$, $-\limsup_{n \to \infty} \frac{1}{n} \log \beta_n$ is maximal.
\end{definition}
Using this definition, we can restrict our focus on tests that rely on the empirical measure of state transitions. This follows from \cite[Lemma~3.5.3]{alma9910123510001591}, whose proof can be readily adapted to accommodate the case of state transition distributions in Markov chains. We use the following theorem to choose the statistic to be used for detection.
\begin{theorem} \label{theorem: optimal test}
    Let test $\set{F}^*$ consist of the maps
    \[
    \set{F}^*(s^n)=
    \begin{cases}
        0 , \text{ if } D_K(\theta, \theta^*)<\eta, \\
        1 , \text{ otherwise.}
    \end{cases}
    \]
    Then $\set{F}^*$ is an optimal test for $\eta$ and its error exponents are given by 
    \begin{align*}
        \lim_{n \to \infty} \frac{1}{n} \log \alpha_n &= -\eta \\
        \lim_{n \to \infty} \frac{1}{n} \log \beta_n &= - \inf_{ \{ \nu : D_K(\nu, \theta^*) < \eta \} } D_K(\nu, \thetaadv) .
    \end{align*}    
\end{theorem}
\begin{proof}
    The proof is provided in Appendix~\ref{app: proof 4 properties}. 
\end{proof}
In this setting, the probability of type I  error is a design choice made by the controller and the adversary cannot affect it. The adversary chooses the adversarial policy in order to maximize the regret while also maximizing the probability of type II error. 
Similarly to the previous settings, the error exponents tend to increase as the $D_{K}$ divergence between $\theta^*$ and $\thetaadv$ increases. In this case, the relevant divergences are the divergence of $\theta^*$ to the decision boundary, which determines the exponent of type I error, and the smallest divergence of $\thetaadv$ to this boundary, which determines the exponent of type II error. Note that it is possible that both $\theta^*$ and $\thetaadv$ are located on the same side of the decision boundary.

Following Theorem~\ref{theorem: optimal test}, assuming that both entities know the maximum acceptable error exponent corresponding to $\alpha_n$, the actuator chooses an acceptable level of the error exponent corresponding to $\beta_n$, calls it $\eta_\beta$, and solves the following optimization problem:
\begin{equation} \label{prob: combined}
\begin{aligned}
    \min_{\pi \in \Pi} \;& \Delta \rho_\pi^\top r \\
    s.t.\; &\inf_{ \{ \nu : D_K(\nu, \theta^*) < \eta \} } D_K(\nu, \thetaadv) <  \eta_\beta~,
\end{aligned}
\end{equation}
where $\rho_\pi$ is the $|\set{S}||\set{A}|$-dimensional vector of space-action frequencies induced by policy $\pi$, $\Pi$ is the set of stationary policies, and $r$ is the $|\set{S}||\set{A}|$-dimensional vector of rewards for all space-action combinations.
Solution to (\ref{prob: infty}) gives the optimal actuator policy given that the detection occurs after observing an infinitely long sequence of states, whereas (\ref{prob: combined}) provides the optimal actuator policy given that the covertness constraints are given in terms of asymptotic exponents of the two types of error. 

\subsection{Discussion}

The following connects the results of the two main sections. The actuator aims to maximize the regret $R(\pi)=J(\pi^*)-J(\pi)$. Let $\pi'$ be the solution of (\ref{prob: infty}). Then, we can write the regret as $R(\pi)=[J(\pi^*)-J(\pi')]+[J(\pi')-J(\pi)]$, where $[J(\pi^*)-J(\pi')]$ represents the regret that can be obtained in a totally covert manner, regardless of the length of the state sequence observed, as this change in policy does not modify the statistics of the induced Markov chain, while $[J(\pi')-J(\pi)]$ is obtained by relaxing the covertness constraint, allowing some information to be `leaked' to the controller, as this change in policy changes the statistics of the induced Markov chain; and thus, this change must be $0$ if detection occurs after observing an infinitely long sequence of states.

Our results show that to achieve a certain degree of covertness, the adversarial cannot use a policy that is `too far' from $\pi^*$ in terms of $D_K$. This means that while the controller may be tempted to simply set $\pi^*$ as the optimal policy for the system, it is also important to consider the neighbourhood of $\pi^*$, as there may be another policy which has lower average rewards, but whose neighbours have better average reward than those of $\pi^*$, meaning that the potential covert policies chosen by the actuator could be less damaging to the system.

\section{Conclusion}

We showed that, given an infinite number of observed states, it is still possible for the adversarial to impact the average reward in certain MDPs without being detected, and the best adversarial policy can be found by solving a linear program. Assuming the number of observed samples is large but finite, the adversary can increase the regret, as the covertness constraint is less tight. We considered error exponents in this setting for adversarial policies known and unknown\textit{ a priori} and showed that they take the form of the divergence $D_K$ between the relevant distributions over state transitions. We plan to expand this work by studying the finite time scenario and the setting in which the adversary does not know the dynamics of the MDP but learns in an online fashion by observing the rewards over time.

\newpage

\bibliographystyle{IEEE}
\bibliography{references}

\newpage

\onecolumn

\appendix

\subsection{Examples of detection at infinity} \label{examples_infinity}

\paragraph{Example 1 - the trivial case} Let $T_s=\begin{bmatrix}
    0.8 & 0.2 \\ 0.5 & 0.5 \\ 0.8 & 0.2
\end{bmatrix}$ and the policy and the reward function in state $s$ be respectively $\pi(\cdot|s)=\begin{bmatrix}
    1 \\ 0 \\ 0
\end{bmatrix}$ and $\pi(\cdot|s)=\begin{bmatrix}
    5 \\ 3 \\ 4
\end{bmatrix}$. As the first and the third row of $T_s$ are equal, it is possible to change the policy in this state from deterministically taking the first action to taking the third, without changing the transition matrix of the induced Markov chain. As the reward of the third action is lower and the transition matrix is unchanged, the adversarial covertly affects the total reward of the system.

\paragraph{Example 2 - covert adversarial behaviour not possible} Let $T_s=\begin{bmatrix}
    0.8 & 0.2 \\ 0.2 & 0.8
\end{bmatrix}$ and the policy in state $s$ be $\pi(\cdot|s)=\begin{bmatrix}
    0.5 \\ 0.5
\end{bmatrix}$. Then we get the constraint $\begin{bmatrix}
    0.8 & 0.2  \\ 0.2 & 0.8 \\ 1 & 1
\end{bmatrix}\begin{bmatrix}
    \Delta \pi(a_0|s) \\ \Delta \pi(a_1|s)
\end{bmatrix} = \begin{bmatrix}
    0 \\ 0 \\ 0
\end{bmatrix}$. It is easy to verify that the rank of $C$ is $2$ and thus the only solution is the trivial $\Delta \pi(\cdot|s)=\begin{bmatrix}
    0 \\ 0
\end{bmatrix}$.

\paragraph{Example 3 - covert adversarial behaviour possible} Let $T_s=\begin{bmatrix}
    0.8 & 0.5\overline{6} & 0.47\overline{3} \\ 
    0.1 & 0.21\overline{6} & 0.26\overline{3} \\ 
    0.1 & 0.21\overline{6} & 0.26\overline{3}
\end{bmatrix}$ and the policy in state $s$ be $\pi(\cdot|s)=\begin{bmatrix}
    0.\overline{3} \\ 0.\overline{3} \\ 0.\overline{3}
\end{bmatrix}$. Then we get the constraint $\begin{bmatrix}
    0.8 & 0.1 & 0.1 \\ 
    0.5\overline{6} & 0.21\overline{6} & 0.21\overline{6} \\ 
    0.47\overline{3} & 0.26\overline{3} & 0.26\overline{3} \\
    1 & 1 & 1
\end{bmatrix}\begin{bmatrix}
    \Delta \pi(a_0|s) \\ \Delta \pi(a_1|s) \\ \Delta \pi(a_2|s)
\end{bmatrix}= \begin{bmatrix}
    0\\0\\0
\end{bmatrix}$. It easy to show that the matrix $C$ in this case has rank$=2$ and that the possible solution are in the form $\Delta\pi(\cdot|s)=\begin{bmatrix}
    0 \\ \alpha \\ -\alpha
\end{bmatrix}$. This gives the adversarial to use any policy $\pi(\cdot|s)=\begin{bmatrix}
    0.\overline{3} \\ 0.\overline{3}+\alpha \\ 0.\overline{3}-\alpha
\end{bmatrix}$. As we are guaranteeing that the distribution of states in the next step is unaffected, it is sufficient to minimize the immediate expected reward, so given $r(s,\cdot)=\begin{bmatrix}
    r_0 \\ r_1 \\ r_2
\end{bmatrix}$, the optimal adversarial policy in state $s$ is $\begin{bmatrix}
    0.\overline{3} \\ 0.\overline{6} \\ 0
\end{bmatrix}$ if $r_1 < r_2$, $\begin{bmatrix}
    0.\overline{3} \\ 0 \\ 0.\overline{6}
\end{bmatrix}$ if $r_1 > r_2$ and any policy on the line satisfying the bounds $0 \leq \pi(\cdot|s)\leq 1$ if $r_1 = r_2$.
\subsection{Derivation of the normalized log-likelihood ratio} \label{app:log-likelihood}

\begin{align*}
    L(s_1,s_2,\dots,s_n)&=\frac{1}{n-1} \log{\frac{P^*_n(s_1,s_2,\dots,s_n )}{\Padv_n(s_1,s_2,\dots,s_n) }} \\
    &=\frac{1}{n-1} \sum_{i=1}^{n-1}\log{\frac{T^*(s_i,s_{i+1})}{\Tadv(s_i,s_{i+1})}} \\
    &=\frac{1}{n-1} \sum_{a_1,a_2 \in \set{S}^2} (n-1)\theta_{s^n}(a_1,a_2)\log{\frac{T^*(a_1,a_2)}{\Tadv(a_1,a_2)}} \\
    &=\sum_{a_1, a_2 \in \set{S}^2} \theta_{s^n}(a_1,a_2) \log{\frac{T^*(a_1,a_2)}{\Tadv(a_1,a_2)}\frac{T_{s^n}(a_1,a_2)}{T_{s^n}(a_1,a_2)}}\\
    &=D_K(\theta_{s^n}, \thetaadv) - D_K(\theta_{s^n}, \theta^*)
\end{align*}

\subsection{Non-existence of isolated points in log-likelihood ratio decision sets} \label{app: no isolated points}

Let $\theta_{x^n},\theta_{y^n} \in \Delta(\set{S}^2)$ be two empirical doublet distributions induced by two different sequences of states $x^n,y^n$. According to Appendix~\ref{app:log-likelihood} we have 

\begin{align*}
    &L(x^n) = \sum_{a_1,a_2 \in \set{S}^2} \theta_{x^n}(a_1,a_2)\log{\frac{T^*(a_1,a_2)}{\Tadv(a_1,a_2)}}, \\&L(y^n)=\sum_{a_1,a_2 \in \set{S}^2} \theta_{y^n}(a_1,a_2)\log{\frac{T^*(a_1,a_2)}{\Tadv(a_1,a_2)}}.
\end{align*}

Let $\theta_{\lambda}=\lambda\theta_{x^n}+(1-\lambda)\theta_{y^n}$, where $\lambda \in [0,1]$. If there exists a sequence of states $s^n_\lambda$ such that $\theta(s^n_\lambda)=\theta_{\lambda}$, then

\begin{align*}
    L(s^n_\lambda)
    &=\sum_{a_1,a_2 \in \set{S}^2}(\lambda\theta_{x^n}(a_1,a_2)+(1-\lambda)\theta_{y^n}(a_1,a_2))\log{\frac{\Tadv(a_1,a_2)}{T_2(a_1,a_2)}}\\
    &=\lambda L(x^n) + (1-\lambda)L(y^n)
\end{align*}

Thus, given a constant $c$, if $L(x^n)>c$ and $L(y^n)>c$, $L(s^n_\lambda)>c$ and conversely, if $L(x^n)<c$ and $L(y^n)<c$, $L(s^n_\lambda)<c$, meaning that the log-likelihood ratio decision sets $B$ and $B^c$ have no isolated points.

\subsection{Proofs of Section~\ref{sec:fixed hypothesis}} \label{app: proof 4 properties}

\begin{lemma} \label{convergence likelihood ratio}
    Let $s^n$ be drawn under $\set{H}^*$. Then
    \begin{align*}
    \frac{1}{n-1}\log{\frac{P^*_n(s^n)}{\Padv_n(s^n)}}
    &= D_K(\theta_{s^n}, \thetaadv) - D_K(\theta_{s^n}, \theta^*) \\
    & \to D_K(\theta^*,\thetaadv) 
    \end{align*}
in probability, where the first equality is due to (\ref{eq:log-likelihood}).
\end{lemma}
\begin{definition}
    For a fixed $n$ and $\delta > 0$, a sequence $s^n \in \set{S}^n$ is said to be $D_K$-typical if and only if 
    \begin{align*}
    D_K(\theta^*,\thetaadv) - \delta
    \leq \frac{1}{n-1}\log{\frac{P^*_n(s^n)}{\Padv_n(s^n)}}
    \leq D_K(\theta^*,\thetaadv) + \delta.
    \end{align*}
\end{definition}
The set of $D_K$ typical sequences is the $D_K$-typical set $A^{(n)}_{\delta}(P^*_n,\Padv_n)$.
\begin{lemma} \label{lemma: 4 properties}
Under the null hypothesis,
\begin{enumerate}
    \item For $s^n=(s_1,s_2,\dots,s_n) \in A^{(n)}_{\delta}(P^*_n,\Padv_n)$,
    \begin{align*}
        P^*_n(s^n)&2^{-(n-1)(D_K(\theta^*,\thetaadv)+\delta)} \leq \Padv_n(s^n)
        \leq P^*_n(s^n)2^{-(n-1)(D_K(\theta^*,\thetaadv)-\delta)}.
    \end{align*}
    \item \label{eq:2nd property} $P^*_n(A^{(n)}_{\delta}(P^*_n,\Padv_n)) > 1-\delta$, for $n$ sufficiently large.
    \item \label{eq:3rd property}$\Padv_n(A^{(n)}_{\delta}(P^*_n,\Padv_n)) < 2^{-(n-1)(D_K(\theta^*,\thetaadv)-\delta)}$.
    \item $\Padv_n(A^{(n)}_{\delta}(P^*_n,\Padv_n)) < (1-\delta)2^{-(n-1)(D_K(\theta^*,\thetaadv)+\delta)}$, for $n$ sufficiently large.
\end{enumerate}
\end{lemma}
\begin{proof}
The first property comes from the definition of $D_K$-typical set. The second property comes from Lemma~\ref{convergence likelihood ratio}. The third property is given by

    \begin{align*}
        \Padv_n(A^{(n)}_{\delta}(P^*_n,\Padv_n))  
        &= \sum_{s^n \in A^{(n)}_{\delta}(P^*_n,\Padv_n)} \Padv_n(s^n) \\
        &\text{by property 1} \\
        &\leq \sum_{s^n \in A^{(n)}_{\delta}(P^*_n,\Padv_n)} P^*_n(s^n)2^{-(n-1)(D_K(\theta^*,\thetaadv)-\delta)} \\
        &= 2^{-(n-1)(D_K(\theta^*,\thetaadv)-\delta)} P^*_n(A^{(n)}_{\delta}(P^*_n,\Padv_n)) \\
        &\leq 2^{-(n-1)(D_K(\theta^*,\thetaadv)-\delta)} .
    \end{align*}

    The last property has a similar proof:

    \begin{align*}
        \Padv_n(A^{(n)}_{\delta}(P^*_n,\Padv_n))  
        &= \sum_{s^n \in A^{(n)}_{\delta}(P^*_n,\Padv_n)} \Padv_n(s^n) \\
        &\text{by property 1} \\
        &\geq \sum_{s^n \in A^{(n)}_{\delta}(P^*_n,\Padv_n)} P^*_n(s^n)2^{-(n-1)(D_K(\theta^*,\thetaadv)+\delta)}  \\
        &= 2^{-(n-1)(D_K(\theta^*,\thetaadv)+\delta)} P^*_n(A^{(n)}_{\delta}(P^*_n,\Padv_n)) \\
        &\text{ by property 2}\\
        &> (1-\delta)2^{-(n-1)(D_K(\theta^*,\thetaadv)\delta)} .
    \end{align*}
\end{proof}
\begin{lemma} \label{lemma: probability other distr}
    Let $B_n \subset \set{X}^n$ be any set of state sequences such that $P^*_n(B_n)>1-\delta$. Let $\Padv_n$ be any other distribution such that $D_K(\theta^*,\thetaadv) < \infty$.
    Then 
    \begin{equation*}
    \Padv_n(B_n) > (1-2\delta)2^{-(n-1)(D_K(\theta^*,\thetaadv)+\delta)}
    \end{equation*}
\end{lemma}
\begin{proof}
    Denote $A^{(n)}_{\delta}(\Padv_n,P_2)$ by $A_n$. Combining $P^*_n(B_n)>1-\delta$ and $P^*_n(A_n)>1-\delta$ from (\ref{eq:2nd property}) gives $P^*_n(A_n \cap B_n)>1-2\delta$. Thus,

    \begin{align*}
        \Padv_n(B_n) &\geq \Padv_n(A_n \cap B_n) \\
        &=\sum_{s^n \in A_n \cap B_n} \Padv_n(s^n) \\
        &\geq \sum_{s^n \in A_n \cap B_n} P^*_n(s^n)2^{-(n-1)(D_K(\theta^*,\thetaadv)+\delta)} \\
        &=2^{-(n-1)(D_K(\theta^*,\thetaadv)+\delta)}\sum_{s^n \in A_n \cap B_n} P^*_n(s^n) \\
        &=2^{-(n-1)(D_K(\theta^*,\thetaadv)+\delta)} P^*_n(A_n \cap B_n) \\
        & \geq 2^{-(n-1)(D_K(\theta^*,\thetaadv)+\delta) }(1-2\delta)
    \end{align*}
\end{proof}

\begin{proof}[Proof of Theorem~\ref{theorem:chernoff}]
    Choose a sequence of sets $\set{A}_n=A^{(n)}_{\varepsilon}(P^*_n,\Padv_n)$, where $0 < \varepsilon \leq \delta$. Lemma~\ref{lemma: 4 properties} states that $P^*_n(\set{A}^c_n)< \varepsilon \leq \delta$ for $n$ large enough, meaning that the relative entropy typical set satisfies the bound on the type I error. The same lemma also shows that
    \begin{equation*}
      \lim_{n \to \infty} \frac{1}{n-1} \log \Padv_n(\set{A}_n) < -D_K(\theta^*,\thetaadv) + \varepsilon,
    \end{equation*}
    providing a lower bound for the error exponent for error type II. Now we show that no other sequence of sets can do better than this bound. Consider any sequence of sets $\set{B}_n$ with $P^*_n(\set{B}_n) > 1-\varepsilon \geq 1-\delta$, then Lemma~\ref{lemma: probability other distr} gives $\Padv_n(\set{B}_n) > (1-2\varepsilon)2^{-(n-1)(D_K(\theta^*,\thetaadv)+\varepsilon)}$, and 
    \begin{align*}
    \lim_{n \to \infty} \frac{1}{n-1} \log \Padv_n(\set{B}_n) 
    > \lim_{n \to \infty} \frac{1}{n-1} \log (1-2\varepsilon) -(D_K(\theta^*,\thetaadv)+\varepsilon) = -D_K(\theta^*,\thetaadv)-\varepsilon.
    \end{align*}
    Combining these two inequalities and letting $\varepsilon \to 0$ proves (\ref{eq: thrm3 main result}).
\end{proof}

\begin{proof}[Proof of Theorem~\ref{theorem: optimal test}]
    By the upper bound of Theorem~\ref{theorem: sanov},

    \begin{align*}
        \lim_{n \to \infty} \frac{1}{n} \log \alpha_n 
        &=
        \lim_{n \to \infty} \frac{1}{n} \log P^*_n(\theta_{s^n} \in \{\nu: D_K(\nu, \theta^*) \geq \eta \} ) \\ &
        = - \inf_{ \{ \nu : D_K(\nu, \theta^*) \geq \eta \}\cap \set{M} } D_K(\nu, \theta^*) \\ &
        = - \eta,
    \end{align*}

    as long as meaning that constraint (\ref{eq: alpha constraint}) is satisfied. We now analyze the probability of error type II under test $S^*$ for the alternative hypothesis for any distribution $\Padv$. By the same upper bound as before we get

    \begin{equation}\label{eq: bound optimal test}
    \begin{aligned}
        \lim_{n \to \infty} \frac{1}{n} \log \beta_n 
        &= \lim_{n \to \infty} \frac{1}{n} \log \Padv_n(\theta_{s^n} \in \{\nu: D_K(\nu, \theta^*) < \eta \}) \\ &
        = - \inf_{ \{ \nu : D_K(\nu, \theta^*) < \eta \} \in \set{M}} D_K(\theta, \theta^\mu) \\ &
        \overset{\Delta}{=} - J(\eta). 
    \end{aligned}
    \end{equation}

    where we let the infimum over an empty set be $\infty$. We then compare these error exponents with those of a test $\set{S}$ determined by the binary function $\set{F}:\Delta(\set{S}^2) \times \mathbb{Z}^+ \mapsto \{0,1\}$ which does not follow the same type of threshold. Suppose that for some $\delta > 0$ and for some $n$, there exists a $\nu \in \set{M}$ such that $D_K(\nu,\theta^*) \leq \eta-\delta$ while $\set{F}(\nu,n)=1$, meaning that the test rejects the null hypothesis even though the distribution of $\eta$ is "closer" to the null hypothesis's distribution in the sense of $D_K$ than the threshold. Then, applying Theorem~\ref{theorem: sanov}, with $\Gamma=\{ \nu \}$ we obtain

    \begin{equation*}
        \lim_{n\to\infty} \frac{1}{n}\log{\alpha_n} = \lim_{n \to \infty} \frac{1}{n} \log P^*_n \{ \theta_{s^n} = \nu \} = - D_K(\nu,\theta^*)
    \end{equation*}


    However, to satisfy constraint (\ref{eq: alpha constraint}) we require that $\eta-\delta \geq \eta$, which contradicts $\delta > 0$. We can conclude that then, for every $\delta > 0$ and for all $n$ large enough,

    \begin{equation*}
    \set{M} \cap \{\nu: D_K(\nu,\theta^*) \leq \eta - \delta \} \subseteq \set{M} \cap \{ \nu: \set{S}(\nu,n)=0 \}
    \end{equation*}

    And so, for every $\delta>0$,

    \begin{align*}
        \limsup_{n\to \infty} \frac{1}{n} \log \beta_n &\geq \liminf_{n\to \infty} \frac{1}{n} \log \beta_n = \\
        &= \liminf_{n\to \infty} \frac{1}{n} \log \Padv_n(\theta_{s^n} \in \{ \nu : \set{S}(\nu,n)=0 \} \cap \set{M}) \\
        &\geq \liminf_{n\to \infty} \frac{1}{n} \log \Padv_n(\theta_{s^n} \in \{ \nu : D_K(\nu,\theta^*) < \eta - \delta \} \cap \set{M}) \\
        &\geq - \inf_{ \{ \nu: D_K(\nu,\theta^*) < \eta - \delta \} \cap \set{M}} D_K(\nu,\thetaadv)
    \end{align*}

    Thus, we can select the $\delta$ providing the tightest bound, which combined with Theorem~\ref{theorem: sanov} gives
    
    \begin{equation}\label{eq: bound general test}
    \begin{aligned}
        \limsup_{n\to \infty} \frac{1}{n} \log \beta_n 
        &\geq -\inf_{\delta > 0} \inf_{ \{ \nu: D_K(\nu,\theta^*) < \eta - \delta \} \cap \set{M}} D_K(\nu,\thetaadv) \\
        &\geq -\inf_{\delta > 0} J(\eta-\delta) \\
        &= -J(\eta) 
    \end{aligned}
    \end{equation}

    which states that (\ref{eq: bound general test}), the exponent achieved by test $\set{F}^*$ is the minimum bound among all possible tests.
\end{proof}

\end{document}